%% file: main.tex
\documentclass[a4paper,10pt,preprintnumbers,floatfix,superscriptaddress,aps,pra
,twocolumn,showpacs,notitlepage,longbibliography,nofootinbib]{revtex4-2}

\usepackage[dvipsnames]{xcolor}
\usepackage[utf8]{inputenc}
\usepackage{amsmath}
\usepackage{amsfonts}
\usepackage{amssymb}
\usepackage{amsthm}
\usepackage{graphicx}
\usepackage{lmodern}
\usepackage{sidecap} 
\usepackage{physics}
\usepackage{mathtools}
\usepackage{dsfont}
\usepackage{lineno}

\usepackage[colorlinks]{hyperref}
\hypersetup{
    colorlinks  = true,
    citecolor   = PineGreen,
    linkcolor   = MidnightBlue,
    urlcolor    = PineGreen
}

\graphicspath{{figures/}{./}}

\newtheorem{lem}{Lemma}

\def\tr{\operatorname{tr}}

\def\ket#1{|#1\rangle}

\def\ketbra#1#2{|#1\rangle\langle #2 |}

\def\c2{\ensuremath{C^{(2)}}}

\begin{document}

\title{Entanglement Detection with Rotationally Covariant Measurements -- \\
From Compton Scattering to Lemonade\\
}
\date{\today}
\author{Marlene Funck}
\affiliation{Institut f\"{u}r Theoretische Physik, Leibniz Universit\"{a}t Hannover, Germany}
\author{Ilija Funk}
\affiliation{Institut f\"{u}r Quantenoptik, Leibniz Universit\"{a}t Hannover, Germany}
\author{Tizian Schmidt}
\affiliation{Institut f\"{u}r Festk\"{o}rperphysik, Leibniz Universit\"{a}t Hannover, Germany}
\author{René Schwonnek}
\affiliation{Institut f\"{u}r Theoretische Physik, Leibniz Universit\"{a}t Hannover, Germany}

\begin{abstract}
\input{abstract}
\end{abstract}
\maketitle

\input{intro}

\section{Rotationally covariant Measurements}
\label{sec2}
We make two basic assumptions: (i)~that the preparation can be faithfully described as a qubit, and (ii)~that the measurement is rotationally covariant (RoC). Considering experiments such as the Compton scattering (Fig.~\ref{fig:rotation_covariant}b), it is clear that the only action, an experimenter takes to switch between different observables, is to rotate some part of the measurement apparatus, for example, a detector (or a whole ring of them). Rotating by any angle $\varphi$ will only cause the measurement data to shift cyclically, without actually changing the statistics. 

Let $\{U_\varphi\}_{\varphi} \subset$ SU(2) be the unitary action of the group of spatial rotations in SO(3) of an angle $\varphi$ acting on the state space $\mathds{C}^2$, i.e. 
\begin{equation}
    U_\varphi = e^{-i\varphi\sigma_y}, \;\;\;\; \varphi \in [0,\pi),
\end{equation}
 where $\sigma_y$ is the Pauli matrix. We choose the $y$-axis to be the rotation axis, because this is the axis of the direction of motion of the photon (the right- and left-handed polarization states are the eigenstates of $\sigma_y$ and stay invariant under rotation). Note that for a spin-1/2 system, one would get a factor of $\frac{1}{2}$ in the exponent and have $\varphi \in [0,2\pi)$ instead, however, photon polarization is not a spin-1/2 system. 

We consider measurements, with outcome set $\Omega=[0,\pi)$ and an absolutely continuous POVM given by $\mathsf{M}[\omega] = \int_\omega M_\varphi d\varphi$ with an operator valued density $M_\varphi$ yielding the probability density of a detection event at $\varphi$.
We call a measurement with such a POVM \textit{rotationally covariant}, if it fulfills
\begin{equation}
\label{POVM}
    M_\varphi =  U^*_{\varphi}M_0U_{\varphi} \; \forall \varphi\in[0,\pi).
\end{equation}
With this, the probability density of a detection event occurring at a differential angle segment $d\varphi$ around $\varphi$ for a given state $\rho$ is
\begin{equation}
    \mathds{P}_\rho(d\varphi) =  \tr[\rho \; U_\varphi^*M_0U_\varphi]d\varphi.
    \label{RoCprobability}
\end{equation}
The probability for an angle interval $[\varphi_1,\varphi_2]$ is obtained by integrating \eqref{RoCprobability} over that interval. By requiring that the total probability is normalized, we obtain, up to a cyclic shift of the angle labels, the following form of $M_0$ (details in \ref{DerivationM0}):
\begin{equation}
    M_0 =\frac{1}{\pi}\;( \mathds{1} + r\sigma_x )\; \; \; \;\;  r\in[0,1],
    \label{M_0}
\end{equation}
where $r$ is a setup-dependent parameter, which we call the `detector contrast'. A `perfect' detector would be obtained for $r=1$, while $r=0$ describes a detector measuring nothing but noise. 

Thus, we have shown that any ideal measurement described by an absolutely continuous normalized POVM that is rotationally covariant (i.e. fulfilling \eqref{POVM}) and that acts on a qubit system is fully characterized by one parameter $r\in[0,1]$ up to a cyclic relabeling of the angles. Then, for every angle $\varphi\in[0,\pi)$, the density is given by equations \eqref{POVM} and \eqref{M_0}. This justifies the notation $M_\varphi^r$ (and $M_0^r$) for our POVM elements. 

As we will see, in the case of Compton polarimeters, $r$ is the `analyzing power'. Given an RoC device, the detector contrast $r$ can be measured. We will see an example for that in Section \ref{sec7}.

\section{Entanglement Detection from Measured Data}
\label{sec3}
Generally, entanglement detection is an NP-hard problem \cite{GURVITS2004448}. For two qubits, however, it is possible and our formalism allows us to provide a straightforward method of determining whether the measured data corresponds to an entangled state.
We do this by reformulating the task as a semidefinite program (SDP) \cite{boyd2004} which is a standard tool in Quantum information \cite{Skrzypczyk_Cavalcanti_Textbook_SDP} and easily solved by, e.g., the MOSEK SDP solver including error estimations. Our reasons for choosing a numerical method are twofold. Firstly, the measurement record is structurally complex, and general binning over arbitrary angle segments prevents a compact analytical description. Secondly, we must incorporate optimization over finite sample errors, which SDPs handle naturally via convex confidence sets.

Assume that we have a bipartite system $\rho\in\mathcal{H}_\text{A}\otimes\mathcal{H}_\text{B}$ with two parties, Alice and Bob, each holding a qubit and an RoC measurement device with a finite resolution. This is defined by a finite angle partition $\{\Delta\varphi_i\}_{i\in I}$ ($I$ being some index set) such that $\cup_{i\in I}\Delta\varphi_i = [0,\pi)$. Then Alice's POVM is the set of operators
\begin{equation}
    A_i = \int_{\Delta\varphi_i}\frac{1}{\pi}U_\varphi^*(\mathds{1} + r_\text{A} \:\sigma_x)U_\varphi \;d\varphi, \;\;\;\; i\in I
\label{povmelmt}
\end{equation}
and Bob's $B_j$ respectively, where the detector contrasts $(r_\text{A}, r_\text{B})$ fully determine their measurement devices. Now the measurement statistic $\{p_{ij}\}_{i,j\in I}$ for Alice detecting an event in $\Delta\varphi_i$ and Bob simultaneously in $\Delta\varphi_j$, reflects the probabilities $P_{ij} = \tr[\rho \, (A_i\otimes B_j)]$ for a given state $\rho$.

We can make a statement about entanglement by solving a minimization problem: Let $\mathcal{D}$ be the set of density matrices for a two-qubit system. Then the state is entangled, if
\begin{equation}
    \inf \Bigl\{\mathcal{N}(\tilde\rho)\, \bigl| \, \tilde\rho \in \mathcal{D}, \,\tr\bigl[\tilde\rho \, (A_i \otimes B_j)\bigr] = p_{ij} \, \forall i,j\Bigr\} > 0.
    \label{min}
\end{equation}
Here, we use the negativity $\mathcal{N}(\rho)=\frac{\|\rho^\Gamma\|_1-1}{2}\geq0$ \cite{vidalWerner2002computablemeasureofentanglement}, where $\rho^\Gamma$ denotes the partial transpose of $\rho$ on Alice's part of the system and $\|\cdot\|_1$ denotes the trace norm. It is a measure of entanglement which by the PPT criterion fulfills $(\mathcal{N}(\rho)>0 \Leftrightarrow \rho \text{ entangled})$ for two-qubit states \cite{peres1996separability, HORODECKI1996, vidalWerner2002computablemeasureofentanglement}. This minimization can be formulated as an SDP (derivation in \ref{derive_sdp}):
\begin{equation}
\begin{split}
    \inf\{&\tr[T]-1 \,|\, T \geq \tilde\rho^\Gamma, \, T \geq 0,  \,
    \\
    &\tilde\rho\geq0, \, \tr[\tilde\rho]=1, \, \tr[\tilde\rho(A_i \otimes B_j)]=p_{ij} \,\forall ij\},
\end{split}
    \label{SDP}
\end{equation}
where the minimization runs over two complex $4\cross 4$ matrices: $T$ and the density matrices $\tilde\rho$ compatible with the probabilities. We can detect entanglement by solving this SDP. Note that for any  $r_\text{A}, r_\text{B}>0$, entanglement detection is possible in theory. However, the statistical deviation of finite samples will scale with $r_\text{A}, r_\text{B}$.   

When dealing with experimental data, there will be errors $\varepsilon_{ij}$, such that the correct probabilities lie within the range $p_{ij}\pm\varepsilon_{ij}$. The SDP from above can still be used, if we change the trace equalities $\tr[\tilde \rho \: A_i\otimes B_j] = p_{ij}$ into inequalities including the errors, namely,
    \begin{equation}
        p_{ij}-\varepsilon_{ij} \leq \tr[\rho \: A_i\otimes B_j] \leq p_{ij}+\varepsilon_{ij} \; \forall ij.
    \end{equation}
  Choosing a confidence level $c\in[0,1)$ and assuming that the measurement data is independently and identically distributed (I.I.D.), we can use Hoeffding's inequality \cite{hoeffding1994probability} to estimate the probability $\mathds{P}$ that the deviation $\varepsilon_{ij} = |p_{ij} - P_{ij}|$ of the experimental probabilities $p_{ij}$ from the actual probabilities $P_{ij}$ is smaller than a $c$-dependent threshold (details in \ref{error_est_derivation}):
\begin{equation}
    \mathds{P}\Biggl(\varepsilon_{ij} < \sqrt{\frac{1}{2n}\ln \biggl(\frac{2\,|I|^2}{1-c}\biggr)} \;\; \forall \, i,j\in I \: \Biggr) \; \geq \; c,
    \label{error}
\end{equation}
where $n$ is the number of measurements. Both the negativity $\mathcal{N}$ and the map $\rho \rightarrow \tr[\rho \, (A_i\otimes B_j)]$ are continuous with respect to the trace norm, so the maximum error in the negativity's minimum will be small for small $\varepsilon_{ij}$.
\section{Semi-Device Independence}
\label{sec4}
It is possible to determine the detector contrast, if we are provided with perfectly prepared states. But what if we analyze top quarks or other systems that cannot be reliably prepared? Can  we detect entanglement if the state and the RoC devices (i.e.~the detector contrasts $r_\text{A}$, $r_\text{B}$) are unknown? Indeed, we can. The following approach can detect entangled states without exact knowledge of the RoC-devices involved. While it may not detect all entangled states, the detections it makes will be accurate.

The procedure is as follows: Assume we have obtained the measurement statistics $\{p_{ij}\}_{i,j}$, where $p_{ij} = \tr[\tilde\rho \, (A^{r_\text{A}}_i\otimes B^{r_\text{B}}_j)]$ for an unknown sample state $\rho$ and unknown contrasts $(r_\text{A}, r_\text{B})$ of our sample devices. Choosing arbitrary hypothetical devices defined by detector contrasts $(r_\text{A}^{\text{hyp}}, r_\text{B}^{\text{hyp}})$, there are three possible cases:
\\[2mm]
1. $\{p_{ij}\}$ can only be replicated with entangled $\rho^\text{hyp}$
\\
2. $\{p_{ij}\}$ can be replicated with some non-entangled $\rho^\text{hyp}$
\\
3. $\{p_{ij}\}$ cannot be replicated with any $\rho^\text{hyp}$,
\\[2mm]
where we say that the statistics are replicated, if $\tr[\rho^\text{hyp}\, (A^{r^\text{hyp}_\text{A}}_i\otimes B^{r^\text{hyp}_\text{B}}_j)] = p_{ij}$ holds for all $i,j$. For each hypothetical device $(r^\text{hyp}_\text{A}, r^\text{hyp}_\text{B})$, we can determine the case we are in by minimizing the SDP \eqref{SDP}: If it returns minimal negativity $\mathcal{N}_{\text{min}} = 0$, we are in case 2. If the minimum negativity is $\mathcal{N}_{\text{min}}>0$, we are in case 1. If it returns nothing, we are in case 3.

If by solving the SDP for all hypothetical devices, $(r^\text{hyp}_\text{A}, r^\text{hyp}_\text{B})$, we find that there is -- even one -- device where case 2 appears, we cannot say with certainty that actual sample state $\rho$ was entangled. If, however, for all hypothetical devices only case 1 and 3 appear, we can deduce that the state $\rho$ that produced the statistics must have been entangled. Solving the SDPs for all hypothetical devices can be done numerically sufficient by, e.g., applying the seesaw method \cite{Werner2001seesaw, pal2010seesaw}. Hence, this procedure allows for semi-device independent entanglement detection.

Assuming a maximally entangled sample state, we can analytically derive the threshold below which semi-device-independent entanglement detection is impossible (details in \ref{minimal_rArB_derivation}):
\begin{equation}
    r_\text{A}r_\text{B}  \leq \frac{1}{3}.
    \label{detect_cond}
\end{equation}
In particular, both sample detector contrasts $r_\text{A}$ and $r_\text{B}$ must be greater than $1/3$ for any successful semi-device independent entanglement detection, which results from the above equation when one of the detector contrasts equals 1. We note that this threshold is not necessarily optimal. A tight threshold for individual cases can be estimated numerically. The result for the maximally entangled case and where both devices have the same detector contrast is presented in Fig.~\ref{fig:device indep}.
\begin{figure}[h!]
    \centering
    \includegraphics[width=0.95 \linewidth]{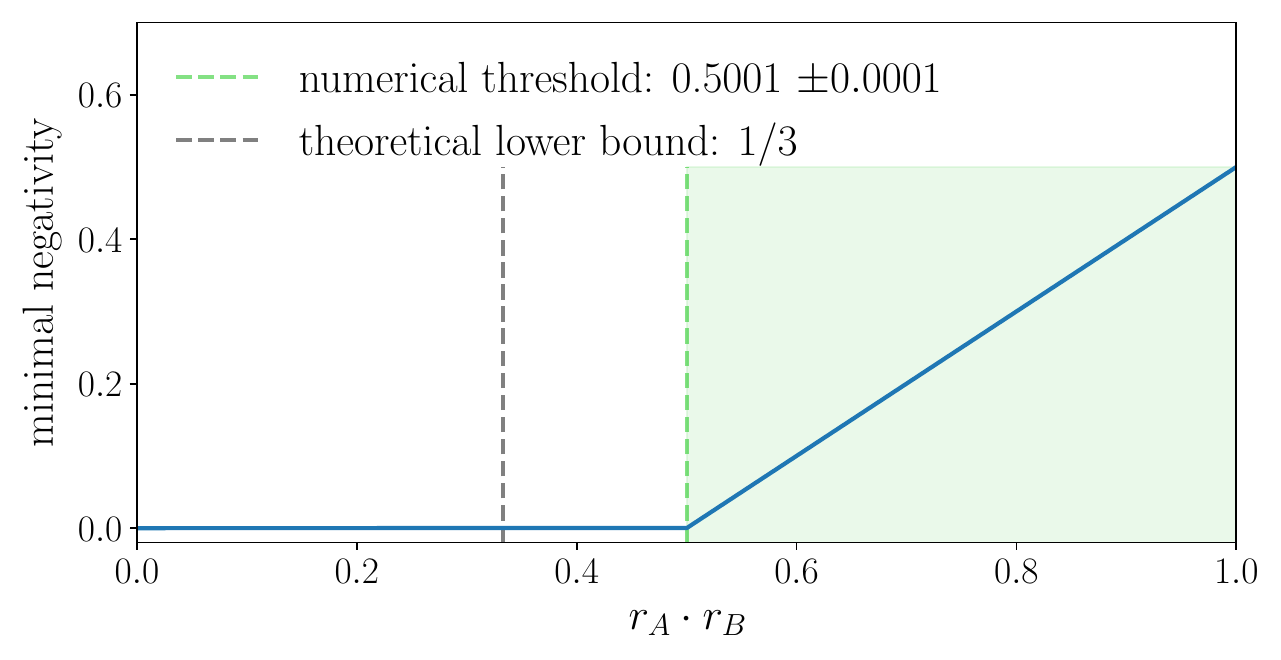}
    \caption{Entanglement detection for equal sample devices $r_\text{A}=r_\text{B}$ in the maximally entangled state $\ketbra{\Psi^-}{\Psi^-}$. The minimal negativities result from solving SDP \eqref{SDP} with perfect hypothetical devices $r^\text{hyp}_\text{A}=r^\text{hyp}_\text{B}=1$, which, in this case, suffice as they yield the relevant results (see \ref{Method_semi_device_indep}). Entanglement detection is possible if the minimized negativity is positive, in this scenario, for $r_\text{A} \cdot r_\text{B}\ge0.5001\pm0.0001$. Intuitively, the worse the detector contrasts, the more likely are false-negative results. More details can be found in \ref{Method_semi_device_indep}.}
    \label{fig:device indep}
\end{figure}

\section{Bell-Experiments and Steering}\label{sec5}
We have established that our setting suffices to detect entanglement. A natural question arising is whether we can also observe steering or detect Bell violations, as these are the next items in the hierarchy \cite{Wiseman_Steering2007}.

First, we show that demonstration of quantum steering with a rotationally covariant device is indeed possible, by providing a theoretical scenario: Assume that Alice prepares a two-qubit state, sending one part of it to Bob, and performing one out of two possible measurements (with two possible outcomes) on her side of the system, while Bob performs one out of two measurements with an RoC device with detector contrast $r$. Alice's task is to (classically) communicate her measurement results to Bob and to convince him, that her state is entangled. Bob is convinced, if he finds that the joint statistics of the two are not reproducible by a local hidden state (LHS) model \cite{Wiseman_Steering2007, cavalcanti2015analogCHSH}.

Cavalcanti et al.~\cite{cavalcanti2015analogCHSH} have derived an analogue to the CHSH inequality for steering: If it is violated, there is no LHS model and steering is demonstrated. If Alice and Bob both perform $\sigma_x$ and $\sigma_z$ measurements, it reads
\begin{equation}
    \begin{split}
        &\sqrt{\langle(\sigma_x^A + \sigma_z^A)\sigma_x^B\rangle^2 + \langle(\sigma_x^A + \sigma_z^A)\sigma_z^B\rangle^2} 
        \\
        &+ \sqrt{\langle(\sigma_x^A - \sigma_z^A)\sigma_x^B\rangle^2 + \langle(\sigma_x^A - \sigma_z^A)\sigma_z^B\rangle^2} \; \leq \; 2,
    \end{split}
\label{CHSH-analog}
\end{equation}
where $\langle\sigma^A_i\sigma^B_j\rangle = \tr[\rho \: (\sigma^A_i \otimes \sigma^B_j)]$ for $i,j\in\{x,z\}$ and $\rho$ is the bipartite state they share. Note that, in their work, they derive it for general dichotomic observables.

As shown in \ref{Methods_pauli_measurements}, Bob can obtain the statistics for $\sigma_x$ and $\sigma_z$ measurements with the help of post-processing, if he knows the exact value of $r$. If Alice (secretly) prepares a Werner state $\rho = f\ketbra{\Psi^-}{\Psi^-} + \frac{1-f}{4}\mathds{1}$ and performs $\sigma_x$ and $\sigma_z$ measurements we have $\langle\sigma^A_i\sigma^B_j\rangle = -f\delta_{ij}$. Then the left-hand side of eq. \eqref{CHSH-analog} is $2\sqrt{2}\,f$, so in the assumed scenario, steering is demonstrated, if $f>1/\sqrt{2}$.

If $r$ is unknown, then Bob can derive the statistics of $r\sigma^B_x$ and $r\sigma^B_z$ measurements (see \ref{Methods_pauli_measurements}). In that case, according to the authors of the steering inequality from above, this factor of $r$ must be simply added to the inequality \cite{cavalcanti2015analogCHSH}. Thus, we obtain $2\sqrt{2}\,fr$ for the left side of \eqref{CHSH-analog}, so steering can be demonstrated even if the exact value of $r$ is unknown, provided that $fr>\frac{1}{\sqrt{2}}$. 

In the case of Compton-scattered annihilation photons, where the opening angle $\theta$ determines $r$, \cite{caradonna2025entanglement} present the angles for which steering is possible using a complementary approach working with the full two-photon density matrices.

To detect Bell violations with two RoC devices, we can select different angle partitions defining Alice's and Bob's respective POVMs as in \eqref{povmelmt}. With two angle partitions into two parts on each side, we obtain the standard CHSH experiment \cite{CHSH1969}. For other numbers of measurement devices (angle partitions) and outcomes, we obtain other Bell-type experiments. $\{M_\varphi^{r_\text{A}}\}_{\varphi \in [0,\pi)}$ and $\{M_\varphi^{r_\text{B}}\}_{\varphi \in [0,\pi)}$ form so-called parent POVMs for all such measurements of Alice and Bob. A simple calculation shows that this is enough to already know that a local hidden variable (LHV) model exists, and hence no Bell violation is possible \cite{Wiseman_Steering2007} (detailed derivation in \ref{Methods_Bell}).

\section{Recovering the Klein--Nishina formula from our formalism}
\label{sec6}
One example of a frequently discussed experiment that fits into our framework is the Compton scattering of linearly polarized photons on free electrons (Fig.~\ref{fig:rotation_covariant}b). Until today, the Klein--Nishina formula \cite{KleinNishina} is used to describe this setting \cite{Wu1950, KasdayWu, Hiesmayr_Moskal_2019, Caradonna_2019, Abdurashitov2022, Ivashkin2023, Watts2021, strizhak2022, Zaidi2004}. Although there are approaches for quantum information descriptions \cite{Caradonna_2019, Hiesmayr_Moskal_2019, Caradonna2025, Caradonna2024stokes_rep}, an explicit formulation through POVMs is still missing. Our formalism can transfer specific descriptions like the Klein--Nishina formula to POVM-descriptions by simply adjusting the detector contrast $r$:

First, we consider the Klein--Nishina cross section for linearly polarized photons scattered from an electron. It is given by \cite{KleinNishina, collaboration2020physics, heitler1948quantum}
\begin{equation}
\begin{split}
    \hspace*{-8pt}\frac{d\sigma}{d\Omega} \hspace*{-1pt}&= \hspace*{-1pt}\frac{1}{2}\; r_0^2 \;\frac{\nu^2}{\nu_0^2}\,\biggl[\frac{\nu_0}{\nu} + \frac{\nu}{\nu_0}-2\sin^2\hspace*{-1pt}\theta\; \cos^2\hspace*{-1pt}\varphi\biggr]
    \\
    \hspace*{-1pt}&= \hspace*{-1pt}\frac{1}{2} \; r_0^2 \; \frac{1}{(1+\beta(1-\cos\theta))^2}
    \\
    &\hspace*{-3pt}\biggl[1\hspace*{-2pt}+\hspace*{-2pt}\beta(1\hspace*{-1pt}-\hspace*{-1pt}\cos\theta)\hspace*{-1pt} + \hspace*{-1pt}\frac{1}{1\hspace*{-2pt}+\hspace*{-2pt}\beta(1\hspace*{-1pt}-\hspace*{-1pt}\cos\theta)}\hspace*{-1pt}-\hspace*{-1pt}2\sin^2\hspace*{-2pt}\theta\cos^2\hspace*{-2pt}\varphi\biggr],
\end{split}
\end{equation}
where $\beta = h_0\nu m_e^{-1}c^{-2}$ is the ratio between the incident photon's and the electron's energy, $r_0$ is the classical electron radius, and $\nu/\nu_0=1/[1+\beta(1-\cos\theta)]$ is the ratio between the scattered and the incident photon's energies. The cross section element then is $d\sigma = d\sigma/d\Omega \sin{\theta}\, d\theta \,d\varphi$. We here use the notation where $\theta$ is the opening angle and $\varphi$ is the angle between the scattering plane and the incident photon's polarization. As we show in \ref{derivation_klein_nishina}, the resulting probabilities for detecting a photon at angle $\varphi$ (fixing $\theta$) match the ones given by our RoC-framework, iff the detector contrast is
\begin{equation}
    r(\beta, \theta) = \frac{\sin^2\theta}{1+\beta(1-\cos\theta)+\frac{1}{1+\beta(1-\cos\theta)} - \sin^2\theta}
    \label{r_roc}
\end{equation}
\\
Thus, we can directly translate between the Klein--Nishina cross section and our RoC-framework. 
\begin{figure}
    \centering
    \includegraphics[width=0.7\linewidth]{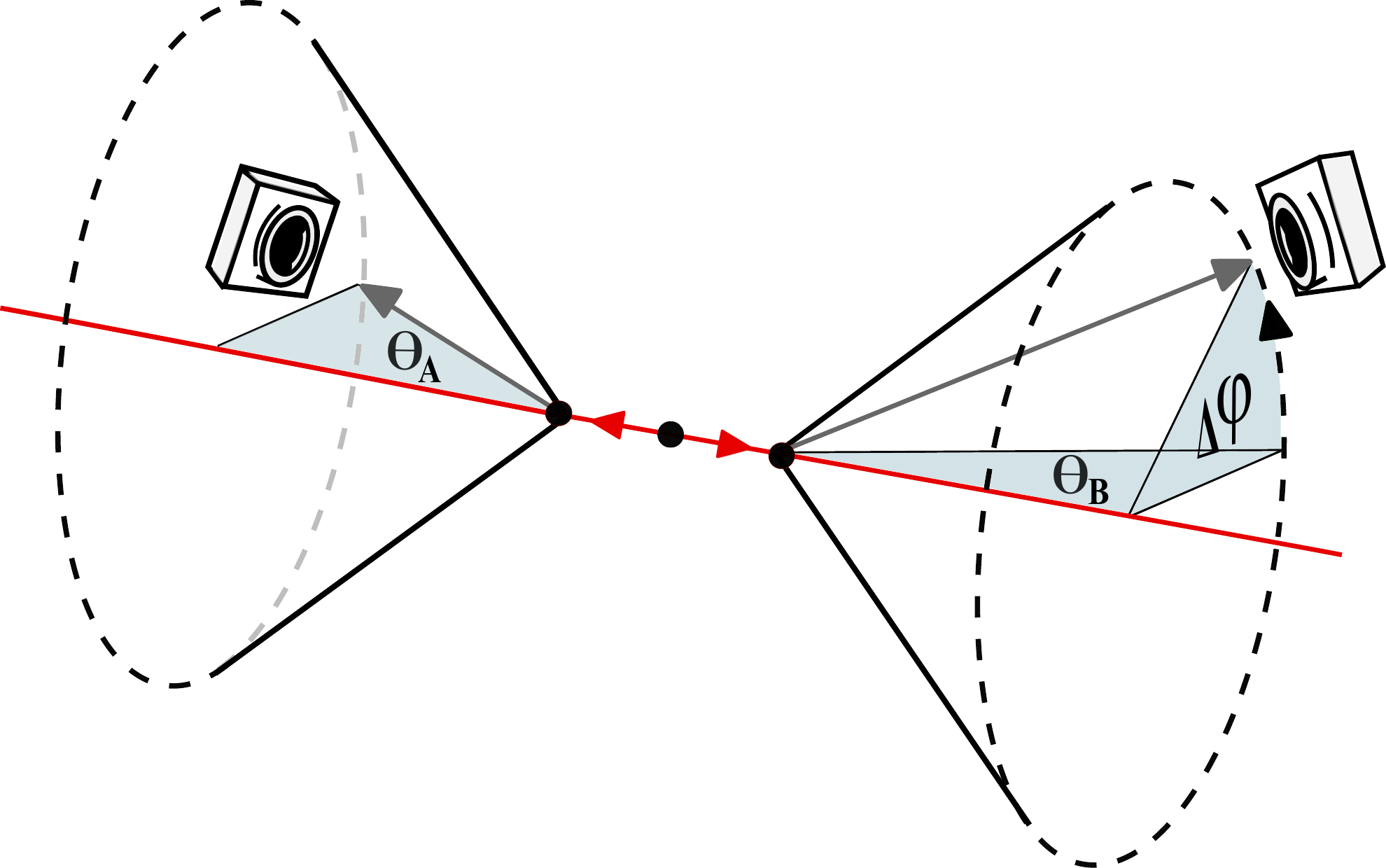}
    \caption{Two Compton polarimeters. A source emits orthogonally linearly polarized (entangled) photons which scatter from electrons. Coincidence events of the photons that are scattered in directions $(\theta_\text{A}, 0)$ and $(\theta_\text{B}, \Delta\varphi)$ are detected.}
    \label{fig:entanglement}
\end{figure}

Secondly, we can recover the probabilities for case of electron-positron annihilation photons (which are orthogonally linearly polarized) being Compton scattered into two directions described by the opening angles $\theta_\text{A}$, $\theta_\text{B}$ and $\Delta\varphi = \varphi_\text{B} - \varphi_\text{A}$, the angle between their scattering planes (see Fig.~\ref{fig:entanglement}). There has been a vast amount of experiments and theoretical work, e.g., about correlations in this process \cite{Pryce1947,KasdayWu, Ivashkin2023, Caradonna_2019, Hiesmayr_Moskal_2019, Watts2021, parashari2024, tkachev2025, moskal2025nonmaximal}. It is usually described by a probability distribution, sometimes called Pryce--Ward equation after \cite{Pryce1947}. When fixing both the photon energy and the opening angles $\theta_i$, renormalization yields
\begin{equation}
\label{normalized_diff_angle}
   \mathds{P}_{\theta_\text{A}, \theta_\text{B}}(\Delta\varphi) = \frac{1}{\pi}[1-\alpha(\theta_\text{A})\alpha(\theta_\text{B})\cos(2\Delta\varphi)].
\end{equation}
$\alpha(\theta_\text{A})$ and $\alpha(\theta_\text{B})$ are the analyzing powers of the Compton polarimeters \cite{Abdurashitov2022}. Defining the orthogonally linearly polarized photon state as the pure state $\rho_\perp$ given by $\frac{1}{\sqrt{2}}(\ket{\text{HV}} - \ket{\text{VH}})$, the RoC framework yields the probability 
\begin{equation}
\begin{split}
    \mathds{P}^{\text{RoC}}_{r_\text{A} r_\text{B}}(\Delta\varphi) &= \int_0^{\pi} \tr\bigl [\rho_\perp \, (M^{r_\text{A}}_{\varphi}\otimes M^{r_\text{B}}_{\varphi + \Delta\varphi})\bigr ] d\varphi\\
    & = \frac{1}{\pi}[1 - r_\text{A}r_\text{B}\cos(2\Delta\varphi)].
\end{split}
\label{P_compt}
\end{equation}
Hence, we can reproduce \eqref{normalized_diff_angle} for $r_i = \alpha(\theta_i)$. According to \cite{Ivashkin2023, Snyder1948} the maximal value of $\alpha(\theta_i)$ (and hence $r_i$) in this particular setting is $0.69$. For known $r$, we can use the entanglement detection scheme in \ref{sec3}. Our analysis from \ref{sec4} shows that using two such Compton polarimeters one cannot perform semi-device independent entanglement detection on the maximally entangled state: For the best-case scenario, we obtain $r_\text{A}r_\text{B}=0.69^2=0.4761$ which is smaller than the threshold we found in Fig.~\ref{fig:device indep}. 

Regarding the conflicting statements of \cite{Hiesmayr_Moskal_2019, Caradonna_2019,Ivashkin2023, Caradonna2024stokes_rep} about the distinguishability of separable and entangled states and the possibility of entanglement detection, \eqref{P_compt} and our results from above are in agreement with the statements in \cite{Caradonna2024stokes_rep}. 

Thirdly, when analyzing the above scattering towards correlations, in many cases the fraction $R$ of perpendicular and parallel scattering events is considered \cite{Abdurashitov2022,Ivashkin2023, KasdayWu, BrunoAgostino1977, wilsonLowe1976, Langhoff1960, Pryce1947, Faraci1974, Caradonna2024stokes_rep}. Again, this can be derived easily with our framework, using \eqref{P_compt}:
\begin{equation}
    R = \frac{P^{\text{RoC}}_{r_\text{A}r_\text{B}}(\pi/2)}{P^{\text{RoC}}_{r_\text{A}r_\text{B}}(0)} = \frac{1+r_\text{A}r_\text{B}}{1-r_\text{A}r_\text{B}}.
\end{equation}

In summary, our formalism is well suited for ongoing discussions, especially about quantum correlations. The methods presented in \ref{sec2}, \ref{sec3} and \ref{sec4}, allow for a simple yet precise analysis of quantum correlations, offering a more comprehensive perspective than case-specific descriptions, which often fail to capture the whole picture.

\section{Soft-Drink-Based detectors}
\label{sec7}
We now introduce a simple showcase experiment based on photon-scattering from a selection of soft drinks that demonstrates how our formalism can be easily applied and how the detector contrast can be determined experimentally. First, note that the polarization of a photon can be described as a qubit, and that linear polarization can be rotated with a half-wave plate (see Fig.~\ref{fig:rotation_covariant}c). 
When being shot at suitable liquids, linearly polarized photons scatter polarization-dependently, which yields a power distribution over the angles $\varphi\in [0,2\pi)$. In accordance with our symmetry assumptions, instead of measuring this distribution with a detector rotating around the scatterer, one can rotate the wave plate (and thus the incident polarization), while keeping the detector at a fixed angle. The performed experiment is visualized in Fig.~\ref{fig:setup}. Note that the photons are emitted by a 30~mW continuous-wave laser with a beam diameter of 5~mm. Although the laser is not a single-photon source, we operate far below any optical nonlinearity: For studying nonlinear effects in liquids, pulsed lasers with femtosecond pulses and peak powers between kilowatts to hundreds of terawatts are typically used \cite{kedenburg2014nonlinear, castner1995intermolecular}. Therefore, we assume that the scattering behavior and hence the determined parameter $r$ can be directly extrapolated to single-photon measurements.
\begin{figure}[h!]
    \centering
    \includegraphics[width=0.9\linewidth]{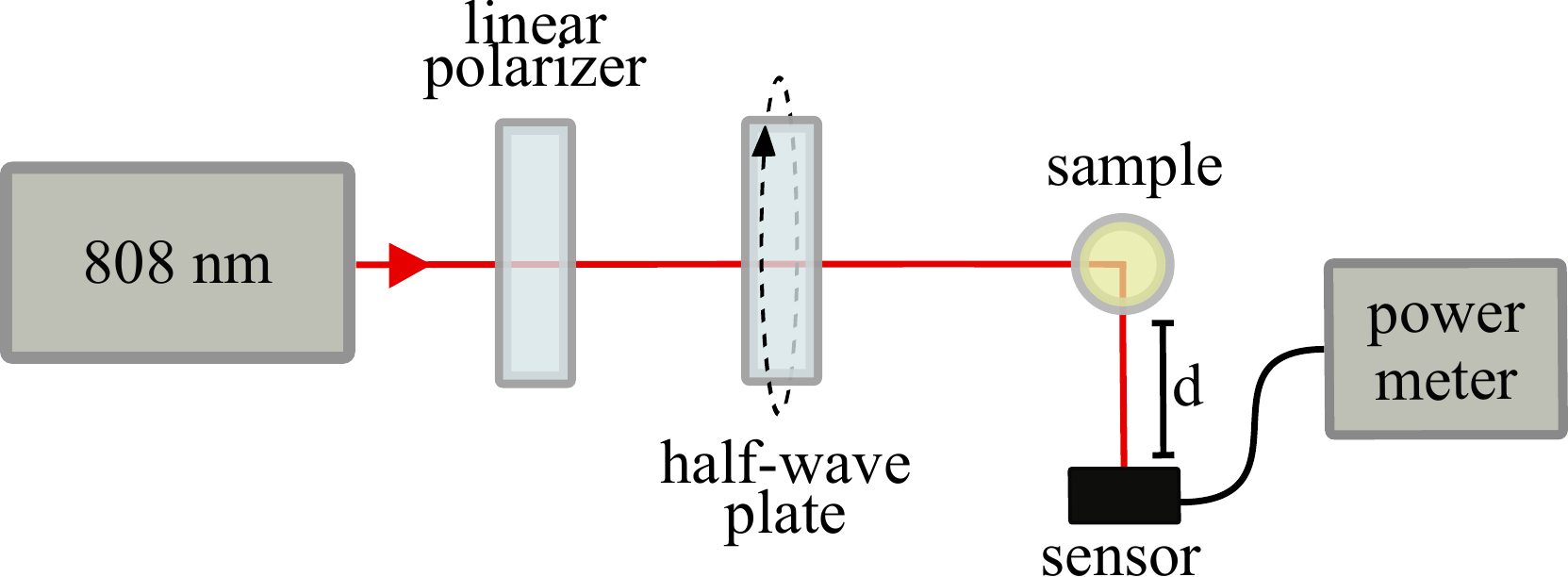}
    \caption{The experiment: A single mode, fiber-coupled, 30~mW power, $808$~nm, continuous-wave diode laser, followed by a linear polarizer (anti-reflection coated for $808$~nm) with an extinction ratio of $10^8$:$1$. The polarization is rotated by a half-wave plate (anti-reflection coated for $808$~nm). The light scatters from different sample liquids contained in a glass test tube (diameter: $25$ mm) and is measured roughly at a right angle compared to the incoming beam using an optical power meter. Powers were measured with an integration time of 1~s.
    }
    \label{fig:setup}
\end{figure}

We can describe the photons as maximally linearly polarized after passing through the linear polarizer. Choosing initial direction in the Bloch-sphere, will only cyclically shift the measurement data, so we can set i.e., $\rho_\text{lin} = (\mathds{1} + \sigma_x)/2$. Since our linear polarizer has an extinction ratio of $10^8:1$, the non-maximally polarized parts will not have a significant impact on our measured statistics. Even if they did, however, our analysis still works the same and the detector contrast for a perfect polarizer would be even higher. The distribution we expect is then given by the RoC-POVM \eqref{POVM}:
\begin{equation}
\mathds{P}^r_{\rho_\text{lin}}(\varphi) = \tr[\rho_\text{lin} \, M_\varphi^r]=\frac{1 + r \cos(2\varphi)}{\pi}.
\label{lemonade_probs}
\end{equation}
By fitting the measurement data to \eqref{lemonade_probs}, we can determine the detector contrast $r$ of this setup. It is influenced by, e.g., the
detector distance $d$, the diameter of the test tube, laser stability, detector efficiency, but also, to a large extent, by the selected sample liquid. The results in (Table \ref{tab:liquid_comparison}) show that, in a fixed setup, we can change the detector contrast within a wide range by changing the sample liquid. These measurements were done for $d=20$~mm, with the dominant contribution to the uncertainty of $r$ arising from the laser intensity stability.
\begin{table}
    \centering
    \begin{tabular}{c|c}
        Liquid & \; Detector contrast $r$ from fit\\
        \hline
        Elderflower syrup & $0.246\pm0.001$\\
        Fanta-lemon & $0.390\pm0.001$\\
        Vitamin water & $0.395\pm0.001$\\
        Maracuja lemonade & $0.489\pm0.001$\\
        Grapefruit lemonade & $0.498\pm0.001$\\
        Apple lemonade & $0.610\pm0.001$\\
    \end{tabular}
    \caption{Changing the liquid sample significantly influences the measured detector contrast $r$. The measurements were performed with a fixed detector distance $d=20$ mm. $r$ was obtained by fitting the data to $f(\varphi) = a(1 + r\cos(2(\varphi + \varphi_0)))$.  Each data point is the mean of 10 measurements. The dominant contribution to the uncertainty of $r$ arises from the laser intensity stability.}
    \label{tab:liquid_comparison}
\end{table}
As a second step, the detector contrast $r$ was further increased by reducing the distance between the sensor and the liquid from $d=20$ mm to $d=15$ mm. The corresponding results, showing the measured scattered output power as a function of the angle, are presented in Fig.~\ref{fig:Power_distr}.

Using apple lemonade as a scatterer, we find that a detector contrast of $r=0.704\pm0.001$ can be achieved, which is higher than the known theoretical upper bound for Compton scattering experiments of electron-positron annihilation photons (\ref{sec6}).
Any characterized detector with $r>0$, can detect entanglement in some states, see \ref{sec3}. The threshold for semi-device independent entanglement detection with equal devices is $r_\text{thresh}=\sqrt{0.5001}\approx 0.707$ (see Fig.~\ref{fig:device indep}), assuming a maximally entangled state, so the found detector contrast would not suffice for this procedure. However, it is possible that this threshold can be reached with an optimized geometry or liquid samples.

The uncertainties obtained here quantify the fit uncertainty dominated by laser-intensity stability. A full analysis of the feasibility of entanglement detection of entangled single photons would additionally require a systematic error budget including all relevant contributions. One way to implement an entanglement detection scheme is to use entangled photon pairs generated by spontaneous parametric down-conversion and hand Alice and Bob each a glass of lemonade. Entanglement sources with the required fidelity are widely available. The measured distribution of coincidence events in all angle segments can then be used to solve the SDP \eqref{SDP}. Of course, achieving a sufficiently high rate of coincidence events such that the experiment is feasible in a reasonable time frame is challenging. Nonetheless, our findings indeed give no reason to believe that entanglement detection with lemonade should be ruled out.  

This simple experiment reveals the core strength of our approach. The complex composition of the molecules, particles and bubbles found in some of the examined liquids, makes it difficult to determine the relative contributions of Rayleigh, Mie and geometric scattering. Furthermore, small changes in setup parameters can have a significant impact on the measured statistics. Our approach circumvents the need of knowledge or control over these parameters and the scattering mechanisms, as it only depends on the realization of the rotational symmetry: We can simply measure the detector contrast, yielding the full description needed for the entanglement analysis.

\begin{figure}[h!]
    \centering
    \includegraphics[width=1\linewidth]{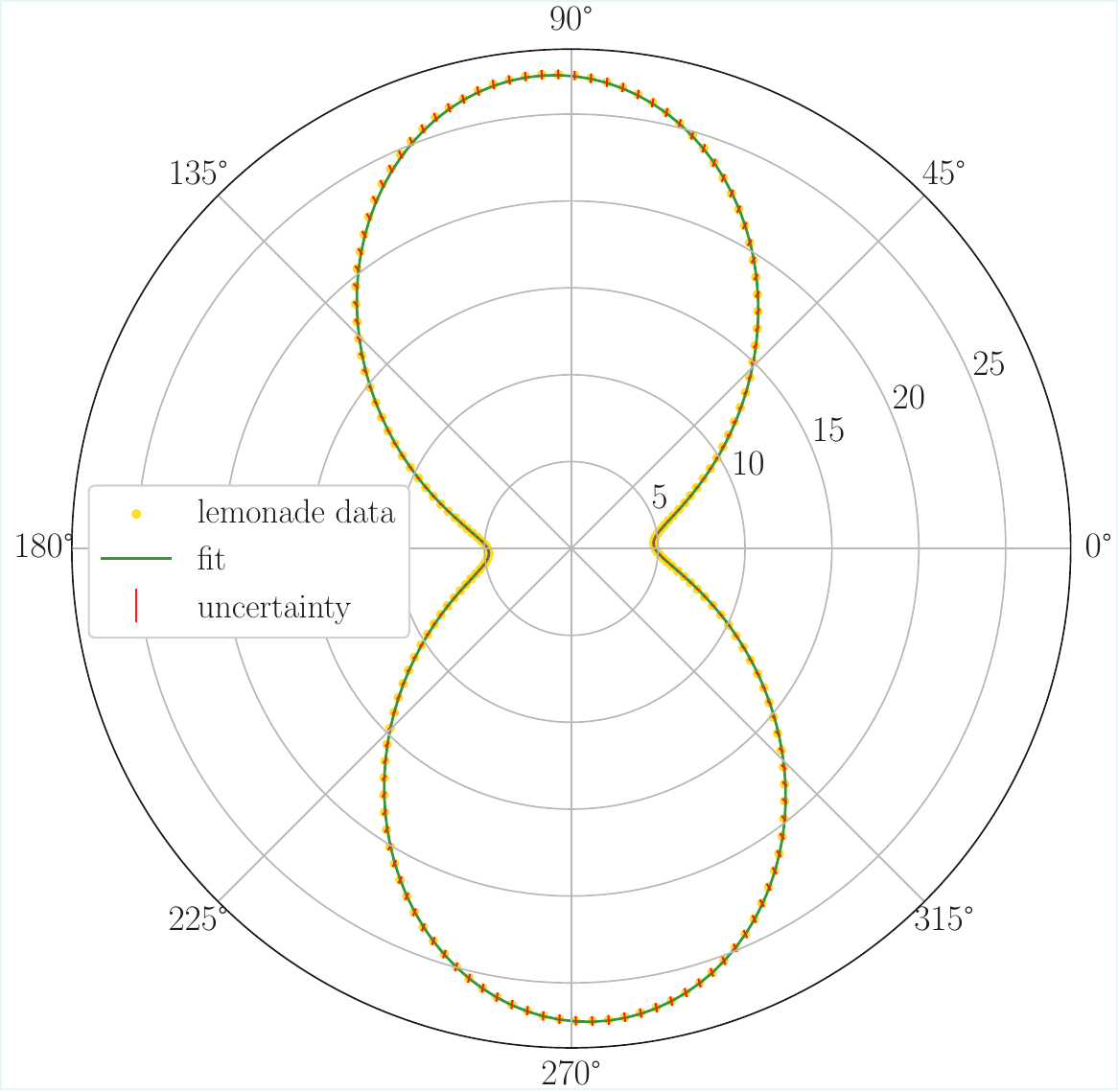}
    \caption{Scattered power distribution in $\mu$W for apple lemonade, resulting in $r=0.704\pm0.001$, measured at a distance of $d=15$ mm. $r$ was obtained by fitting the data to $f(\varphi) = a(1 + r\cos(2(\varphi + \varphi_0)))$.  Each data point is the mean of 10 measurements. The dominant contribution to the uncertainty of $r$ arises from the laser intensity stability. The data is taken for angles in $[0,2\pi)$ instead of $[0,\pi)$, which is why we obtain a mirror-symmetric plot.}
    \label{fig:Power_distr}
\end{figure}

\section*{Conclusion}
Our work opens the door to a clean, quantum information analysis of rotationally covariant single photon experiments. In particular, we derived the POVM-formulation of the Klein\text{--}Nishina formula, describing Compton scattering of polarized photons. In fact, we found that idealized RoC measurements are indistinguishable from a quantum information perspective, as they are fully classified by the established detector contrast $r$, determining the quality of the measurement. Furthermore, we have introduced an entanglement certification method for RoC measurements, only requiring the measured data and an SDP solver, and we have generalized this method for unknown detector contrasts $r$. Lastly, we found that any combination of measurements obtained from the same two RoC setups admits a parent POVM and is not suitable for Bell violations, while steering with one characterized RoC measurement is possible.

One direction for future work is the analysis of U(1) representations beyond qubits, by analyzing U(1) symmetries on multi-photon states and on arbitrary systems. Another direction is to extend this analysis to detectors that admit an O(3) covariance. Those are, for example, natural for scattering experiments, in which the detector is rotated on two angles (see Fig.~\ref{fig:rotation_covariant}b). Additionally, one could test the predictions of the Klein--Nishina formula with our framework. Lastly, our hope is to enable entanglement detection with our introduced lemonade-experiment.

\subsection*{Acknowledgments}
We gratefully acknowledge discussions with Otfried Gühne, Mathias Kleinmann, Robert Raußendorf, and Tobias Osborne. M.F. acknowledges helpful comments by Lukas Hantzko, Lennart Binkowski, and Hauke Hinrichs. M.F. and R.S. acknowledge support by the Quantum Valley Lower Saxony, the Cluster of Excellence Quantum Frontiers, and by the BMBF projects ATIQ, QUICs, SEQUIN and CBQD.

%

\section{Methods}
\subsection{Derivation of $M_0$}
\label{DerivationM0}
We derive \eqref{M_0}: Dealing with a probability distribution, we assume
\begin{equation}
    \begin{split}
        1 = \mathds{P}_\rho([0,\pi)) = &\int_{0}^{\pi} \;\tr[\rho \; U^*_{\varphi}M_0U_{\varphi}]\;d\varphi \\
        = &\;\tr[M_0\int_0^{\pi} U_{\varphi} \rho \:U^*_{\varphi} \; d\varphi]\\
        = &\pi\;\tr[M_0 \; \bar\rho],
    \end{split}
\end{equation}
where $\bar\rho = \frac{1}{2}(\mathds{1}+\alpha\,\sigma_y)$ is the density matrix averaged over all possible rotations along the Bloch's sphere's $y$-axis ($\alpha\in[-1,1]$). This must hold for any state $\rho$, so
\begin{equation}
    1 = \frac{\pi}{2}\tr[M_0(\mathds{1}+\alpha\:\sigma_y)] \;\;\; \forall \alpha\in[-1,1]
\label{cond1}
\end{equation}
Since $M_0$ is positive (and hermitian), it can be written in the form $M_0=\lambda\mathds{1}+\vec k \vec\sigma$ for some real $\lambda$ and $\vec k$ with $|\vec k|^2 = k_x^2+k_y^2+k_z^2\le \lambda^2$. Using $\tr[\sigma_j]=0$ and $\sigma_i^2=\mathds{1}\;\forall j$, \eqref{cond1} becomes $1 = \frac{\pi}{2}(2\lambda + 2\alpha k_y)\;\forall\alpha\in[-1,1]$, which immediately yields $k_y = 0$, $\lambda=1/\pi$ and $k_x^2+k_z^2\le1/\pi^2$. Thus,
\begin{equation}
    M_0 = \frac{\mathds{1} + r (\cos{\theta} \: \sigma_x + \sin{\theta} \: \sigma_z )}{\pi}
\end{equation}
for some setup-dependent $r\in[0,1]$ and $\theta\in[0,2\pi)$.
A simple calculation shows that $M_\varphi^\theta = M_{\varphi+\theta/2}^0$ holds for any $\varphi$ and $\theta$, so the initial rotation labeled with the angle $0$ can be chosen arbitrarily (it simply changes which angle is labeled with 0). We can set $\theta=0$ which yields the assertion.

\subsection{Deriving the SDP Formulation}
\label{derive_sdp}
We show that \eqref{min} can be formulated as the SDP in \eqref{SDP}: Any hermitian operator $A$ can be decomposed into two positive semidefinite operators $A_+$ and $A_-$ such that $A = A_+ - A_-$ and $A_+A_- = A_-A_+ = 0$. Since the trace norm equals the sum of the eigenvalues' absolute values, we obtain $$\|\rho^\Gamma\|_1 = \tr[\rho^\Gamma_+] + \tr[\rho^\Gamma_-]=2 \, \tr[\rho^\Gamma_+] - \tr[\rho^\Gamma] = 2 \, \tr[\rho^\Gamma_+] - 1.$$
Hence, $\mathcal{N}(\rho) = \tr[\rho^\Gamma_+]-1$. By applying Lemma \ref{tr+_lemma}, we immediately obtain the form in \eqref{SDP}.

\subsection{Error Estimation for RoC Measurements}
\label{error_est_derivation}
Assume that we perform $n$ measurements in an RoC setting yielding I.I.D. random variables $X_1, X_2, ..., X_n$. This means assuming that our measurement device does not change over time and that the different shots do not influence each other. Under this assumption, we can apply Hoeffding's inequality \cite{hoeffding1994probability}, namely 
\begin{equation}
    \mathds{P}(|\sum_{i=1}^n X_i - \mu | \ge \varepsilon) \le 2 \exp{-\frac{2\varepsilon^2}{\sum_{i=1}^n(b_i-a_i)^2}} \; \forall \epsilon>0,
    \label{hoeffding}
\end{equation}
where $\mu$ is the expectation value of the underlying
probability distribution, and $a_i \le X_i \le b_i$ almost surely (i.e.~the probability that $X_i$ lies in this range is 1).

Having simultaneously performed RoC measurements at angle segment $\Delta \varphi_i$ on Alice's qubit and at $\Delta \varphi_j$ on Bob's qubit, the probability for a detection event is given by $P_{ij} = \tr[\rho \, (A_i\otimes B_j)]$. Measuring $n$ times, we obtain a list $(X_1, ..., X_n)$, where $X_k=1/n$ if the $k$'th measurement detected an event in the chosen angle segments of Alice and Bob and $X_k = 0$ otherwise. Hence, the measured frequency is $p_{ij} = \sum_{k=1}^nX_k$ and in the above notation, we have $a_k = 0$ and $b_k=1/n$ for all $k$. Therefore, \eqref{hoeffding} becomes
\begin{equation}
    \mathds{P}(|p_{ij} - P_{ij}| \ge \varepsilon) \le 2e^{-2n\varepsilon^2} \;\;\; \forall \epsilon > 0.
    \label{bound_probs}
\end{equation}
If we want the bound to hold for all bins $(i,j)$ at once, we can estimate the bound by
\begin{equation}
\begin{split}
    \mathds{P}(\exists\,&(i,j)\hspace*{-2pt}\in I\hspace*{-2pt}\cross\hspace*{-2pt} I\hspace*{-2pt}: |p_{ij} - P_{ij}| \ge \varepsilon)
    \\
    &\leq \sum_{i,j} \mathds{P}(|p_{ij} - P_{ij}| \ge \varepsilon) \; \leq \; |I|^2 \; 2e^{-2n\varepsilon^2}
\end{split}
\end{equation}
from this it follows directly that
\begin{equation}
    \mathds{P}(|p_{ij} - P_{ij}| < \varepsilon \; \forall \, i,j) \; \geq \; 1 - 2 \, |I|^2 \, e^{-2\, n \, \varepsilon^2} \equiv c,
\end{equation}
where we call $c\in[0,1)$ the confidence level. Solving the last equality for $\varepsilon$ yields $\sqrt{\frac{1}{2n}\ln \frac{2\,|I|^2}{1-c}}$, i.e.
\begin{equation}
    \mathds{P}\Biggl ( |p_{ij}- P_{ij}| < \sqrt{\frac{1}{2n}\ln \biggl(\frac{2\,|I|^2}{1-c}\biggl)} \; \forall i,j \: \Biggr ) \geq c.
    \label{prob_bound}
\end{equation}

\subsection{Example for Semi-device Independent Entanglement detection}
\label{Method_semi_device_indep}
Fig.~\ref{fig:device indep} shows one example of entanglement detection, where we assume that we have identical sample devices $r_\text{A} = r_\text{B}$. The full detection scheme, as explained in \ref{sec5}, considers all hypothetical devices and not only the optimal one.
The statistics $\{p_{ij}\}$ are simulated with the maximally entangled sample state $\ketbra{\Psi^-}{\Psi^-}$ and for several sample devices with detector contrast products $r_\text{A}r_\text{B}\in[0,1]$. For each of these statistics, the minimal negativities for respective hypothetical detector contrasts are obtained from solving SDP \eqref{SDP} (using 'MOSEK'-solver in CVXPY). Thereby, for each sample device we find the minimal negativity for every hypothetical device. This is minimal for $r_\text{A}^\text{hyp}r_\text{B}^\text{hyp} = 1$, so it suffices to simply consider these optimal hypothetical devices, as was done in \ref{fig:device indep}.

\subsection{Semi-device independent threshold for $(r_\text{A},r_\text{B})$}
\label{minimal_rArB_derivation}
The maximally entangled state $\ketbra{\Psi^-}{\Psi^-}$ is the Werner state $W_f = f\ketbra{\Psi^-}{\Psi^-} + \frac{1-f}{4} \mathds{1}$ with $f=1$. The probability of a detection event in angles $\varphi$ and $\psi$ is
\begin{equation}
    \tr[W_{f}(M_\varphi^{r_\text{A}} \otimes M_\psi^{r_\text{B}})] = \frac{1 - f\:r_\text{A}r_\text{B}\cos[2\,(\varphi-\psi)]}{\pi^2},
    \label{Wernerprobs}
\end{equation}
where we use the POVM elements as defined in \eqref{POVM}.
This shows that having a Werner state with fidelity $f$ and a measurement device defined by $(r_\text{A}, r_\text{B})$ is equivalent to having an ideal measurement device (i.e.~$r_\text{A}=r_\text{B}=1$) and a Werner state with fidelity $\tilde f = f\:r_\text{A}r_\text{B}$. 

A Werner state is separable, iff its fidelity $\tilde f \leq 1/3$. Hence, given a maximally entangled state ($f=1$) and the statistics, it is impossible to detect entanglement, if
\begin{equation}
    r_\text{A}r_\text{B}  \leq \frac{1}{3}.
\end{equation}
\subsection{Pauli-Measurements with RoC-Devices}
\label{Methods_pauli_measurements}
Let Bob's RoC device be defined by detector contrast $r$. We can define two different measurements (POVMs) using \eqref{povmelmt}: By partitioning the angles horizontally (i.e.~$\{\Delta\varphi_i^z\} = \{[0,\pi/2],[\pi/2, \pi]\}$) and vertically (i.e.~$\{\Delta\varphi_i^x\} = \{[\pi/4, 3\pi/4], [0,\pi/4]\cup[3\pi/4, \pi]\}$), we obtain the operators
\begin{equation}
    R^z = \int_0^{\pi/2} M_\varphi^r \, d\varphi,  \;\;\;\; S^z = \int_{\pi/2}^{\pi} M_\varphi^r \, d\varphi
\end{equation}
and $R^x, \:S^x$ respectively. A simple calculation shows that $\frac{\pi}{2r}(S^x-R^x) = \sigma_x$ and $\frac{\pi}{2r}(R^z-S^z) = \sigma_z$. Thus, if the value of $r$ is known, post-processing enables measurements of $\sigma_x$ and $\sigma_z$ (knowing the statistics $\tr[\rho \, S^z]$ and $\tr[\rho \, R^z]$, we can calculate the statistics $\tr[\rho \, \sigma_z]$, for $\sigma_x$ analogously). If the detector contrast $r$ is unknown, one cannot divide by its value, so post-processing is limited to statistics $\tr[\rho \, r\sigma_x]$ and $\tr[\rho \, r\sigma_z]$

\subsection{RoC POVMs have LHV Models}
\label{Methods_Bell}
Assume we have the Bell measurements defined by the following POVMs acting on Alice's and Bob's part of the system respectively:
\begin{equation}
\begin{split}
\{A^i_x\}_x &= \{\int_{\varphi \in \Delta\varphi^i_x}M_\varphi^{r_\text{A}} \; | \; x \in X_i\} \\
\{B^j_y\}_y &= \{\int_{\varphi \in \Delta\psi^j_y}M_\varphi^{r_\text{B}} \; | \; y \in Y_j\}, 
\end{split}
\label{Alice'sPOVMs}
\end{equation}
where $X_i$ and $Y_j$ are index sets and $\cup_{x\in X_i} \Delta\varphi^i_x = \cup_{y\in Y_j} \Delta\psi^j_y = [0, \pi)$. Note that the typical CHSH setting is obtained for $i,j\in\{1,2\}$ and $X_i = Y_j = \{0, 1\}$ by assigning $+1$ to the first measurement outcome and $-1$ to the second outcome, for both POVMs, respectively.  
The combined measurement of any pair $(i,j)$ of these POVMs on a bipartite state $\rho$ yields the probabilities
\begin{equation}
    \begin{split}
     &\hspace*{-5pt}\mathds{P}(x,y\;|\; i, j) = \tr[\rho \, (A^i_x \otimes B^j_y)] \\
     &\hspace*{-5pt}= \int_{[0,\pi)} \hspace*{-12pt} d\varphi \hspace*{-3pt} \int_{[0,\pi)} \hspace*{-12pt} d\psi \; \chi_{\Delta\varphi_x^i}(\varphi)\; \chi_{\Delta\psi_y^j}(\psi) \tr[\rho (M_\varphi^{r_\text{A}} \hspace*{-1pt}\otimes\hspace*{-1pt} M_\psi^{r_\text{B}})].
\end{split}
\label{LHV_for_ROC}
\end{equation}
With the hidden variable $\lambda := (\varphi, \psi)$, $\chi_{\Delta\varphi_x^i}(\lambda):=\chi_{\Delta\varphi_x^i}(\varphi)$, and $\chi_{\Delta\psi_y^j}(\lambda):=\chi_{\Delta\psi_y^j}(\psi)$, the above becomes
\begin{equation}
     \hspace*{-5pt}\mathds{P}(x,y\;|\; i, j) = \hspace*{-2pt}\int_{[0,\pi)^2} \hspace*{-2pt} d\lambda  \; \chi_{\Delta\varphi_x^i}(\lambda) \; \chi_{\Delta\psi_y^j}(\lambda) \; \tr[\rho M_\lambda].
     \label{LHV_for_ROC2}
\end{equation}
$\mathds{P}(\lambda) := \tr[\rho M_\lambda] = \tr[\rho\,(M_\varphi^{r_\text{A}} \otimes M_\psi^{r_\text{B}})]$ forms a probability distribution over $\lambda\in[0,\pi)\cross[0,\pi)$. Furthermore, $\sum_{x\in X_i} \chi_{\Delta\varphi_x^i} = \sum_{y\in Y_j} \chi_{\Delta\psi_y^j} = \chi_{[0,\pi)} = 1$, so $\mathds{P}(x \, |\, i,\lambda):= \chi_{\Delta \varphi_x^i}(\lambda)$ and $\mathds{P}(y \, |\, j,\lambda):= \chi_{\Delta \psi_y^j}(\lambda)$ form probability distributions over $x\in X_i$ and $y\in Y_j$, respectively. Hence, equation \eqref{LHV_for_ROC2} is a Local hidden variable model (compare \cite{Wiseman_Steering2007}), which implies that these measurements cannot violate any Bell inequality \cite{Wiseman_Steering2007}.

\subsection{Comparison to Klein--Nishina}
\label{derivation_klein_nishina}
Fixing the opening angle $\theta$ and the energy quotient $\beta$ in the Klein--Nishina cross section, the probability for finding a photon in angle $\varphi$ is
\begin{equation}
\begin{split}
     \mathds{P}_{\beta,\theta}(\varphi) &= \frac{\frac{d\sigma}{d\Omega}(\theta, \varphi)\sin{\theta} \, d\theta}{\int_0^{\pi}\frac{d\sigma}{d\Omega}(\theta, \gamma)\, \sin{\theta} \, d\theta\;d\gamma} 
     \\[2mm]
     &= \frac{1-r(\beta, \theta)\cos(2\varphi)}{\pi},
\end{split}
\label{prob_kn}
\end{equation}
where $r(\beta, \theta)$ is defined as in \eqref{r_roc}.

On the other hand, for a maximally linearly polarized state $\rho = \frac{1}{2}(\mathds{1} - \sigma_x)$, which is the setting of the Klein--Nishina formula, our RoC description yields the probability
\begin{equation}
    \mathds{P}_{r}^{RoC}(\varphi) = \tr[\rho \, U_{-\varphi}M_0U_\varphi] = \frac{1-r\cos(2\varphi)}{\pi}.
\label{prob_roc}
\end{equation}
Solving the equation $\mathds{P}_r^{RoC}(\varphi) = \mathds{P}_{\beta, \theta}(\varphi) \; \forall \; \varphi\in[0,\pi)$ for $r$, where $\theta\in(0,\pi)$
yields equation exactly equation \eqref{r_roc}.
\\[2mm]
\newpage
\appendix
\section{Appendix}
\begin{lem}
    Let A be a hermitian operator on a finite dimensional Hilbert space $\mathcal{H}$ and $\mathcal{T} = \{T:\mathcal{H}\rightarrow\mathcal{H} | \; T\ge0, \; T\ge A\}$. Then  
    $$\tr[A_+] = \inf_{T \in \mathcal{T}}\tr[T]$$.
    \label{tr+_lemma}
\end{lem}
\begin{proof}
    If $A = A_+$ we know that $A$ is positive and the statement is trivial. Assume that $A_+ > A$. We have $A_+\in\mathcal{T}$, because $A_+\ge A$ and $A_+\ge 0$. Hence, it is left to show that there cannot exist a $T \in \mathcal{T}$ s.t. $\tr[T]<\tr[A_+]$. Assume there was such a $T$.
        
    Let $\vec a_1, ..., \vec a_n$ be an orthonormal basis of eigenvectors of $A$ with real eigenvalues $\alpha_1, ..., \alpha_n$ and let it be ordered such that $\alpha_i>0$ for $i\le n_+$ and $\alpha_i\le0$ for $i > n_+$. Since $T\ge A$, for all $i\le n_+$ we have 
    \begin{equation}
        0\le \vec a_i^* \, (T-A) \, \vec a_i = \vec a_i^* \, (T-A_+) \, \vec a_i = \vec a_i^* \, T \, \vec a_i - \alpha_i
        \label{proof_infimum}
    \end{equation}
    Now since $A_+$ is the part of $A$ carrying the positive eigenvalues, we have $A_+\vec a_i = 0$ for all $i\ge n_+$. Therefore,
    \begin{equation}
        \begin{split}
            0 &> \tr[T-A_+] 
            \\&= \sum_{i=1}^{n_+}\vec a_i^* \, (T-A_+) \, \vec a_i 
            + \sum_{i=n_+ + 1}^n \, \vec a_i^* \, (T-A) \, \vec a_i \\
            &= \sum_{i=1}^{n_+} \vec a_i^* \, T \, \vec a_i - \alpha_i
            + \sum_{i=n_+ + 1}^n \, \vec a_i^* \, T \, \vec a_i \; \ge \; 0 + 0,
        \end{split}        
    \end{equation}
    where the last inequality comes from \eqref{proof_infimum} and the fact that $T\ge0$. This contradiction yields the assertion.
\end{proof}
\onecolumngrid
\vspace{8cm}
\centering{Entanglement hides\\
in circular symmetries\\
lemonade finds.}
\end{document}

%% file: abstract.tex
The accurate and efficient detection of quantum entanglement remains a central challenge in quantum information science. In this work, we study the detection of entanglement of polarized photons for measurement devices that are specified by rotational covariance. 
We derive explicit positive operator-valued measures (POVMs) showing that any such (idealized) setting is classified by one real measurable parameter $r$. In particular, we give a POVM formulation of the Klein--Nishina formula for Compton scattering of polarized photons.
We provide an SDP-based entanglement certification method with computable bounds that operates on the full measured statistics, also considering semi-device independent scenarios.
Furthermore, we show that, while Bell violations are impossible with two rotationally covariant measurements, EPR steering can still be certified under one-sided symmetry constraints. Finally, we demonstrate how our approach can be applied by analyzing a rotationally covariant showcase experiment using scattered polarized light in a selection of soft drinks. Our results motivate further investigation of entanglement detection with lemonade-based detectors.

%% file: intro.tex
\section{Introduction}
Entanglement is the central phenomenon that separates quantum theory from a classical description of our world. Beyond this fundamental significance, it is also an essential resource for quantum technologies. This not only includes widely popular fields like quantum computing and cryptography, but also a less often highlighted medical application, positron emission tomography (PET) \cite{Caradonna_2019, Caradonna2024stokes_rep, Caradonna2025, Hiesmayr_Moskal_2019, Watts2021, Abdurashitov2022, Zaidi2004, Romanchek2023, Romanchek2024, Kim2023, Eslami2024, moskal2025nonmaximal, mcnamara2014}, which is one of the possible practical applications of the theory developed in this work.

Tackling the task of detecting and quantifying entanglement in a real-world experiment, one faces a fundamental challenge of theoretical physics: finding a model that is mathematically manageable and enables nontrivial conclusions, while not diverging (at least not significantly) from reality. 
For entanglement mediated by polarized photons, those models span a wide spectrum, ranging from minimal assumptions, like device-independent tests of Bell inequalities \cite{CHSH1969}, to assuming idealized pure states and projective measurements \cite{cohen2019quantum}. While the first comes with very demanding technological requirements for any positive result \cite{Giustina2013,guhne2004detecting,Giustina2015,Shalm2015,Hensen2015}, the latter will almost certainly fail to describe realistic experiments with unavoidable imperfections. 

In this work, we develop and study a model at a useful middle ground: detecting entanglement of two polarized photons using (ideal) measurement devices specified by rotational covariance, which can be found (or created) in a variety of settings (see Fig.~\ref{fig:rotation_covariant}). 
\begin{figure}[h!]
    \centering \includegraphics[width=0.37\textwidth]{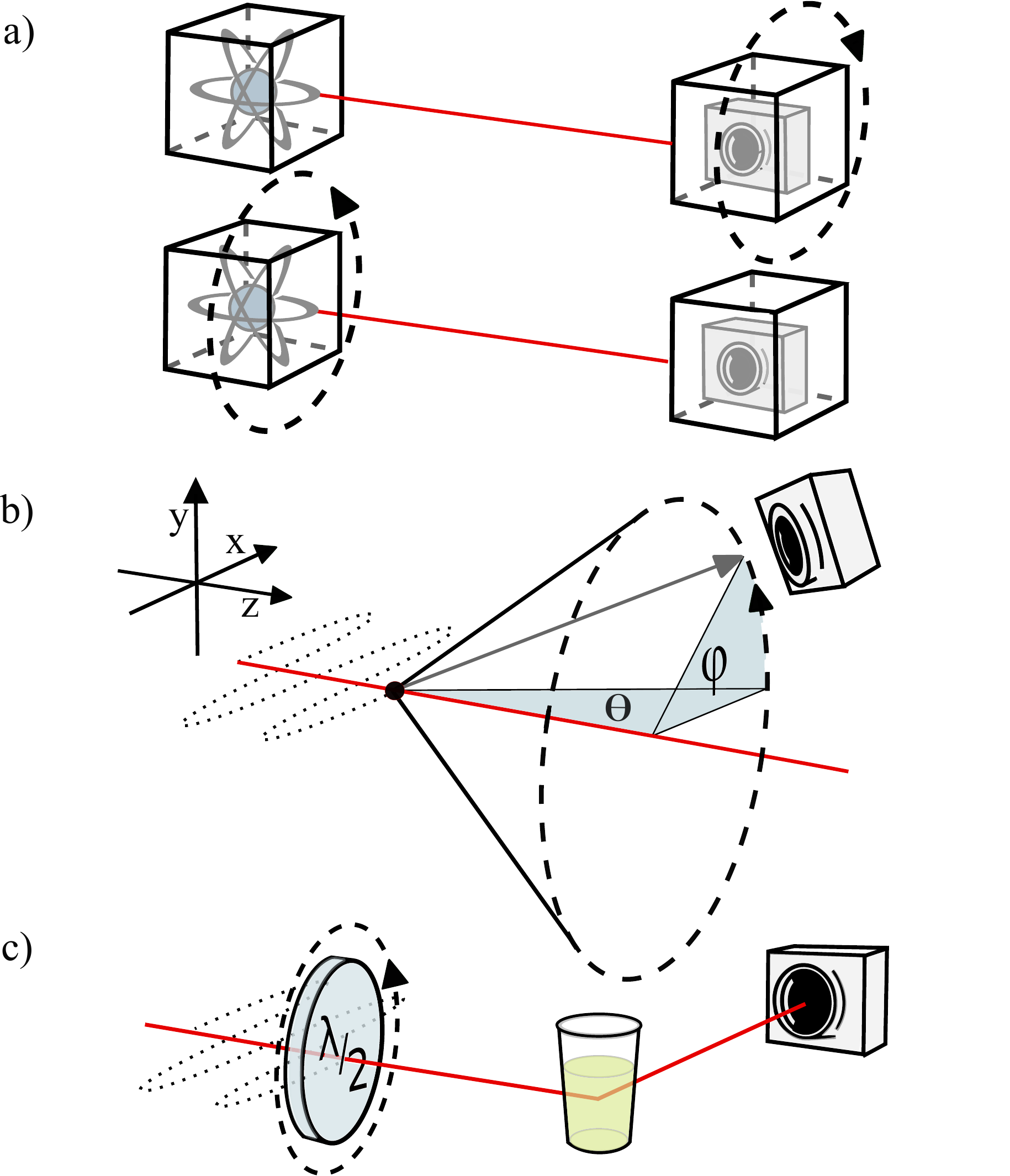}
    \caption{a) Rotationally covariant measurement. Measurement data is $\pi$-invariant. Rotating the preparation system (i.e.~process before detection) or the detector are equivalent in description. b) Compton scattering process. Polarized photons scatter from electrons into angles, $\theta$ and $\varphi$. It is rotationally covariant w.r.t. $\varphi$. c) The lemonade experiment. The polarization of photons is rotated by a half-wave plate. Photons scattered in the lemonade are detected at a fixed angle.}
    \label{fig:rotation_covariant}
\end{figure}
The main motivation of this work stems from Compton scattering experiments, used to detect correlations in $\gamma$-quanta. This mechanism  was used in the very first quantum correlation experiments in the late 1940s \cite{Wu1950} and is today used to enhance PET scanners \cite{Caradonna_2019, Caradonna2024stokes_rep, Caradonna2025, Hiesmayr_Moskal_2019, Watts2021, Abdurashitov2022, Zaidi2004, Romanchek2023, Romanchek2024, Kim2023, Eslami2024, moskal2025nonmaximal}. Beyond this, rotational covariance can also be found in the mechanism that underpins the recent verification of entanglement at the LHC  \cite{Afik2021entanglementtopquarks}. 

The descriptions typically used for entanglement analysis in Compton scattering are based on the Klein--Nishina formula \cite{KleinNishina} which was derived in 1929 and is not easily compatible with the modern quantum information language. It has, e.g., led to conflicting statements about whether or not entangled and separable photon states are distinguishable in Compton scattering setups and whether entanglement detection is possible \cite{Hiesmayr_Moskal_2019, Caradonna_2019, Ivashkin2023, Caradonna2024stokes_rep}. 
To avoid such confusions in the future and simplify the current research in the field, our goal is to provide a precise quantum information description of the generalized setting and provide a comprehensive analysis of entanglement certification therein.

Since highly radioactive material or sufficiently large hadron colliders are rarely found in a standard laboratory for quantum optics, we introduce an easy-to-perform  experiment analyzing scattering of photons with optical wavelengths (see Fig.~\ref{fig:rotation_covariant}c and Fig.~\ref{fig:setup}). We measure the polarization dependent scattering distributions of a selection of soft drinks, and analyze them as rotational covariant detectors in our framework. 
Our results suggest that apple lemonade might be suitable for detecting entanglement. In fact, the detection contrast obtained using lemonade-based detectors exceeds the known limit for Compton scattering of electron--positron annihilation photons.
A fully conclusive experiment, however, demands entangled photon pairs and high sample sizes, which is why we leave it for a future publication.

Beyond serving as a neat showcase and potential student experiment, studying detectors based on soft drinks illustrates a conceptually central aspect of our work: It is unrealistic to find an exact model for a glass of lemonade, reflecting all relevant specifics and imperfections. Nevertheless, by reducing our description to a macroscopically implemented rotational symmetry, we are able to make a comprehensive analysis of the setup.
\section{Background and Setting}
Historically, the investigation of entanglement originates in the well-known Gedankenexperiment of Einstein, Podolski, Rosen \cite{EPR1935} and Schrödinger \cite{Schroedinger1935}, which predicted that quantum theory allows for correlations that cannot be explained classically. Already in 1950, Wu et al.~\cite{Wu1950} performed an experiment on Compton scattering of annihilation photons, which was recognized to be a proof of these correlations by Bohm and Aharonov \cite{Bohm_Aharonov_1957} in 1957. Although their findings predated the more widely known Bell-type (CHSH: \cite{CHSH1969}) photon experiments by over a decade, they are rarely discussed in textbooks. One reason might be that the quantum information description of their setting is more subtle. In contrast to the CHSH setting, the Compton scattering measurement cannot be described by projective measurements typically used for discussing correlation tests. In fact, they are a prototype example of a joint measurement, where the traditional textbook postulate of a projection onto eigenstates fails.
However, the language of quantum information has been developed further over the past decades \cite{Werner1989_Wernerstate, ludwig_2012, kraus1983} and provides the right tool for describing this setting: Quantum states are modeled by density operators $\rho$, which are positive unit trace operators acting on a Hilbert space $\mathcal{H}$. Measurement devices with outcomes in a measurable outcome set $\Omega$ are modeled by positive operator-valued measures (POVMs), which assign a positive operator $\mathsf{M}[\omega]$ to every measurable subset $\omega\subseteq\Omega$ with normalization $\mathsf{M}[\Omega]=1$ such that $\omega \mapsto \tr[\rho \, \mathsf{M}[\omega]]$ is a probability measure for every state $\rho$ \cite{heinosaari2011}. 

In the case of a bipartite setting, the Hilbert space factorizes: $\mathcal{H}=\mathcal{H}_\text{A}\otimes\mathcal{H}_\text{B}$. Local measurements $\mathsf{M}^\text{AB}$ are  modeled  by a product of POVMs with outcomes $\Omega_\text{AB}= \Omega_\text{A} \times\Omega_\text{B}$ and a decomposition $\mathsf{M}^\text{AB}=\mathsf{M}^\text{A}\otimes \mathsf{M}^\text{B}$.  A state $\rho_\text{AB}$ is called separable, if it is a convex combination of product states \cite{Werner1989_Wernerstate}, i.e. 
$\rho_\text{AB}=\sum_i p_i \rho_\text{A}^i\otimes\rho_\text{B}^i$,
and entangled otherwise. 
The task of entanglement detection  concerns the following prototypical question: 
\begin{center}
Given measurement data $\Omega_\text{meas}$ observed on an unknown state $\rho_\text{AB}$, can we find a separable state consistent with our observation -- or -- can we conclude that $\rho_\text{AB}$ is entangled? 
\end{center}
Recently, there has been renewed interest in analyzing entanglement in Compton scattering experiments, driven by its applications in medical imaging with positron emission tomography (PET) scanners \cite{Caradonna_2019, Caradonna2024stokes_rep, Caradonna2025, Hiesmayr_Moskal_2019, Watts2021, Abdurashitov2022, Zaidi2004, Romanchek2023, Romanchek2024, Kim2023, Eslami2024, moskal2025nonmaximal, romanchek2024investigation}. 
Hiesmayr et al.~and Caradonna et al.~\cite{Caradonna_2019, Caradonna2024stokes_rep, Caradonna2025, caradonna2025entanglement, Hiesmayr_Moskal_2019} have proposed ways to describe the process from a quantum information perspective by reformulating the Klein--Nishina formula \cite{KleinNishina}. However, a precise description in quantum information language that rigorously distinguishes between the set of states (i.e., the preparation procedure) and the measurement devices (which are described by POVMs) was still lacking. This has caused confusion about the detectability of entanglement in previous works \cite{Caradonna_2019, Hiesmayr_Moskal_2019, Ivashkin2023}.
In this work, we derive and study the explicit POVMs for any setting with a rotational symmetry (in particular Compton scattering) and present an entanglement detection scheme that can be easily applied to them.

Our approach is based on two fundamental assumptions about the system and the measurement device:\\
(i)  \hspace*{0,02mm} The quantum state of the individual system can be \\
\hspace*{4,4mm} described as a qubit. \\
(ii) The system and measurement device are rotationally \\
\hspace*{4,5mm} covariant (Fig.~\ref{fig:rotation_covariant}) with respect to  U(1).\\
In real-world experiments, assumption (i) is justified whenever the carrier of quantum information is a two-level system, transforming under the symmetry group SU(2), or, in particular, the polarization of a single photon. When implementing this, e.g.~with an SPDC source, this assumption translates as excluding the statistical influence of potential multi-photon events. This is reasonable as long as pump laser powers are kept sufficiently small so that the observed statistics are dominated by single-photon events. Assumption (ii) will be specified in Section \ref{sec2}, where we present the POVM description of a general rotationally covariant (RoC) measurement. Next, in Section \ref{sec3}, we give a clear description of an entanglement detection procedure, including error handling. Section \ref{sec4} elaborates on the semi-device independence of this approach, while Section \ref{sec5} explains why this type of arrangement is unsuitable for Bell experiments, but it can demonstrate steering. In Section \ref{sec6}, we use our formalism to recover well established descriptions of Compton scattering of polarized photons, like the Klein--Nishina formula \cite{KleinNishina}. Finally, in Section \ref{sec7}, we introduce a simple experimental setup, described by our framework, using a laser, a half-wave plate, and a glass of lemonade.